\documentclass[12pt,reqno, letterpaper]{amsart}
\usepackage{amssymb,epsf}
\usepackage{amstext,amsmath,amsfonts,amscd,amsthm,graphicx}
\usepackage{epsfig}
\usepackage{eufrak}
\usepackage{latexsym}
\usepackage{eucal}
\theoremstyle{amsart}
\newfont{\fnt}{cmsy10}
\newfont{\sss}{cmr10}
\newfont{\azb}{wncyr10}
\newfont{\azbit}{wncyi10}
\theoremstyle{definition}

\theoremstyle{plain}
\newtheorem{conj}{Conjecture}
\newtheorem{vt}{Theorem}
\newtheorem{lm}{Lemma}
\newtheorem{ds}{Corollary}
\newtheorem{prp}{Proposition}
\theoremstyle{definition}
\newtheorem{pz}{Remark}
\newtheorem{pr}{Example}

\begin{document}
\title[The Darboux coordinates for a new family of Hamiltonian operators]{
{\protect\vspace*{-1cm}}The Darboux coordinates for a new family of Hamiltonian operators
and linearization of associated evolution equations}
\author{Ji\v{r}ina Vodov\'a}
\address{Mathematical Institute, Silesian University in Opava, Na Rybn\'{i}\v{c}ku \nolinebreak 1, 746 01 Opava, Czech Republic}
\email{Jirina.Vodova@math.slu.cz}
\maketitle
\begin{abstract} {\protect\vspace*{-0.7cm}}
A. de Sole, V. G. Kac, and M. Wakimoto have recently introduced a new family of compatible Hamiltonian operators of the form $H^{(N,0)}=D^2\circ((1/u)\circ D)^{2n}\circ D$, where $N=2n+3$, $n=0,1,2,\dots$, $u$ is the dependent variable and $D$ is the total derivative with respect to the independent variable. We present a differential substitution that reduces any linear combination of these operators to an operator with constant coefficients and linearizes any evolution equation which is bi-Hamiltonian with respect to a pair of any nontrivial linear combinations of the operators $H^{(N,0)}$. We also give the Darboux coordinates for $H^{(N,0)}$ for any odd $N\geqslant 3$. 
\looseness=-1
\end{abstract}
\section{Introduction}
The Hamiltonian evolution equations are well known to play an important role in modern mathematical physics. Indeed, a Hamiltonian operator maps the variational derivatives of the conserved quantities into symmetries; this plays an important role in the theory of integrable systems which often turn out to be \textit{bi}-Hamiltonian, see e.g.\ \cite{dickey, dorfman, fokas, olver} and references therein.
It is thus no wonder that the study and, in particular, the classification of Hamiltonian operators is a subject of ongoing interest, see for instance \cite{astashov,cooke,dubrovin,dorfman,mokhov, mokhov2} and the works cited there.
\looseness=-1

Recently A. de Sole, V. G. Kac, and M. Wakimoto  have made a major advance in this area. Namely, in \cite{kac} they gave a conjectural classification of Poisson vertex algebras in one differential variable, i.e., of scalar Hamiltonian operators. \textit{Inter alia}, they have come up with a new infinite family of compatible Hamiltonian operators $H^{(N,0)}=D^2\circ((1/u)\circ D)^{2n}\circ D$, where $N=2n+3$, $n=0,1,2,\dots$, and $D$ denotes the total derivative with respect to the space variable.\looseness=-1

Although it is well known \cite{olver2, olver4} that there is no proper counterpart of the Darboux theorem on canonical forms of finite-dimensional Poisson structures for Hamiltonian operators associated with evolutionary PDEs, it is often possible to find new variables in which a Hamiltonian operator takes a simpler form. One such form is the Gardner operator $D$; in analogy with the finite-dimensional case the associated new variables are often called the Darboux coordinates, see e.g.\ \cite{cooke,olver2}. Bringing a Hamiltonian operator into the Gardner form enables one e.g.\ to render the associated Hamiltonian systems into the canonical Hamiltonian form and construct Lagrangian representations (modulo potentialization) for  these systems \cite{olver4}.\looseness=-1

In this paper we present the transformations that bring the operators $H^{(N,0)}$ into the Gardner form, see Corollary \ref{ds1} below.  Moreover, in Theorem \ref{th1} we give a differential substitution which \textit{simultaneously} turns the operators $H^{(N,0)}$ for all odd $N\geqslant 3$ into the operators with constant coefficients $\widetilde{H}^{(N,0)}=-D^{2n+1}$.\looseness=-1

These results could be employed e.g.\ for the study of (co)homology of the Poisson complexes associated with the operators $H^{(N,0)}$. The cohomologies in question play an important role e.g.\ in finding all Hamiltonian operators compatible with a given Hamiltonian operator and the associated multi-Hamiltonian systems, see for example \cite{krasilshchik, olver3, sergyeyev} and references therein.
Another possible application of the results in question is e.g.\ the construction of new integrable systems in spirit of \cite{fokas2} in the new variables from Theorem \ref{th1} or Corollary \ref{ds1} with the subsequent pullback to the original variables.
Last but not least, the differential substitution from Theorem~\ref{th1} linearizes any evolution equation which is bi-Hamiltonian with respect to a pair of any nontrivial linear combinations of the operators $H^{(N,0)}$, thus exhibiting a broad class of somewhat unusual (in that they are $C$-integrable rather than $S$-integrable) integrable bi-Hamiltonian systems; see Corollary \ref{ds2} for details. \looseness=-1

\section{Preliminaries}

In what follows we deal with Hamiltonian operators and associated Hamiltonian evolution equations involving a single spatial variable $x$ and a single dependent variable $u$. An
evolution equation of this kind has the form
$$u_t=K[u]=K(x,u,u_x,u_{xx},\dots),$$
where the square brackets indicate that $K$ is a \textit{differential function} in the sense of \cite{olver}, meaning that it depends on $x$, $u$, and finitely many derivatives of $u$ with respect to the space variable $x$. Recall (see e.g.\ \cite{dickey, dorfman, olver} for details) that an evolution equation is said to be \textit{Hamiltonian} with respect to the Hamiltonian operator $\mathcal{D}$ if it can be written in the form $$u_t=\mathcal{D}\delta_u\mathcal{T}[u],$$
where $\mathcal{T}=\int T[u]dx$ is the Hamiltonian functional, and $\delta_u$ denotes the variational derivative with respect to $u$.\looseness=-1

\begin{lm}[\cite{mokhov}]\label{lm1} Let $L_1$ be a Hamiltonian operator in the variables $x, u$.
Under the transformation
\begin{equation}\label{1}x=\varphi(y,v,v_y,\dots,v_{m}),\quad u=\psi(y,v,v_y,\dots,v_{n}),\end{equation}
where $v_{j}=D_y^j(v)$, where $D_y$ is the
total derivative with respect to $y$,
the operator $L_1$ goes into the Hamiltonian operator $L_2$ defined by the formula
\begin{equation}\overline{L}_1=(D_y(\varphi))^{-1}K^*\circ L_2\circ K,\end{equation}
where $\overline{L}_1$ is obtained from $L_1$ under the substitution (\ref{1}) and upon setting $D_x=(D_y(\varphi))^{-1}D_y$,
$$K=\sum_{i=0}^{\max(m,n)}(-1)^iD_y^i\circ\left(\frac{\partial\psi}{\partial v_{i}}D_y(\varphi)
-\frac{\partial\varphi}
{\partial v_{i}}D_y(\psi)\right),$$ 
and $K^*$ is the formal adjoint of $K$.
\end{lm}
\begin{pz}\label{pz1}
Note that in general the operator $L_2$ may contain nonlocal terms unless (\ref{1}) is a contact transformation, cf.\ e.g.\ \cite{astashov, kac, mokhov}.
\end{pz}
\section{The main result}
\begin{vt}\label{th1}
The transformation
$x=v$, $u=1/v_y$ turns
the $N${\rm th} order Hamiltonian operator $H^{(N,0)}=D_x^2\circ((1/u)\circ D_x)^{2n}\circ D_x$, where $N=2n+3$, $n=0,1,2,\dots$, into the Hamiltonian operator {\em with constant coefficients} $\widetilde{H}^{(N,0)}=-D_y^{2n+1}$. \end{vt}
The proof is obtained by a straightforward application of Lemma \ref{lm1}.
Even though the transformation from
Theorem~\ref{th1} is not contact, in the particular case under study
the transformed operators $\widetilde{H}^{(N,0)}$ happen to be free of nonlocal terms (cf.\  Remark \ref{pz1}).

Using Lemma \ref{lm1} we can further amplify the result of Theorem \ref{th1} by providing the Darboux coordinates (cf.\ Introduction) for the operators $H^{(N,0)}$.\looseness=-1
\begin{ds}\label{ds1}
The transformation
$x=(-1)^{\frac{n+1}{2}}w_{n}$, $u=(-1)^{\frac{n+1}{2}}/w_{n+1}$, where $w_{k}=D_z^ k(w)$, and $z$ is the new independent variable, maps
the Hamiltonian operator $H^{(N,0)}=D_x^2\circ((1/u)\circ D_x)^{2n}\circ D_x$ to the first-order Gardner operator $D_z$ for any odd $N\geqslant 3$.
\end{ds}

\begin{pz}\label{pz2}
The inverse of the transformation $x=v$, $u=1/v_y$ is nothing but the extended hodograph transformation in the sense of \cite{clarkson}. Let us stress that, unlike the original transformation, the said inverse is not a differential substitution, i.e., it cannot be written in the form $y=\xi(x,u,u_x,\dots,u_q)$, $v=\chi(x,u,u_x,\dots, u_p)$.
\end{pz}
 In \cite{kac} the following conjecture is stated:
\begin{conj}[De Sole, Kac, Wakimoto]
For any translation-invariant Hamiltonian operator $H$ of order $N\geq 7$ there exists a contact transformation that brings $H$ to either a quasiconstant coefficient skew-adjoint differential operator, or to a linear combination of the operators $H^{(j,0)}$ with $3\leq j\leq N$, $j$ odd. \end{conj}
Recall that a differential function is called quasiconstant \cite{kac} if it depends only on $x$.

If we further allow for the transformation of the form
$x=v, u=1/v_y$ from Theorem \ref{th1} and use the latter, 
we can state a somewhat stronger conjecture:
\begin{conj}\label{conj1}
Any translation-invariant Hamiltonian operator $H$ of order $N \geq 7$  
can be transformed into a quasiconstant coefficient skew-adjoint differential operator
using either a contact transformation or a composition thereof with the transformation $x=v, u=1/v_y$.
\end{conj}
The transformation $x=v, u=1/v_y$ can be written as a composition of the potentiation $x=z$, $u=w_z$ and of the hodograph transformation 
$z=v$, $w=y$. The latter is a contact (in fact, even a point) transformation, 
so, with the obvious change of notation, we can recast the above conjecture as follows:\looseness=-1
\begin{conj}\label{conj2}
Any translation-invariant Hamiltonian operator $H$ of order $N \geq 7$
can be transformed into a quasiconstant coefficient skew-adjoint differential operator
using either  a contact transformation or a composition thereof with 
the potentiation $x=y, u=v_y$.
\end{conj}
Thus, assuming that the above conjecture holds true, 
extending the class of allowed equivalence transformation using potentiation enables 
us to have just one type of normal forms for  translation-invariant Hamiltonian operators, namely, the skew-adjoint differential operators with coefficients that depend only on $x$.  
 
Upon further allowing for the transformations of the form
$x=y, v=\sum_{i=0}^k b_i(x)v_i$ and using the following
\begin{lm}
Let $H=\sum_{i=0}^N a_i(x)D_x^i$ be quasiconstant skew-adjoint differential operator. There exists a transformation of the form
$x=y,v=\sum_{i=0}^k b_i(x)v_i$ that turns $H$ to the Gardner operator $D_y$.
\end{lm}
\noindent
we arrive at a yet stronger conjecture:
\begin{conj}\label{conj3}
Any translation-invariant Hamiltonian operator of order $N \geq 7$ can be transformed 
into the Gardner operator $D$ by either a transformation from Conjecture~\ref{conj2}
or a composition thereof with a transformation of the form $x=y, v=\sum_{i=0}^k a_i(x)v_i$.
\end{conj}

\section{Applications to evolution equations}
Recall that the transformation $x=v$, $u=1/v_y$ can be written as the composition of the potentiation
$x=z$, $u=w_z$ and of the hodograph transformation $z=v$, $w=y$. The first of these is nothing but introduction of the potential $w$ for $u$. It is readily seen that a bi-Hamiltonian evolution equation
\begin{equation}\label{3}u_t=H_1\delta_u\mathcal{T}_1=H_2\delta_u\mathcal{T}_2,
\end{equation} where $H_i=\sum_{j=1}^{k_i}c_{ij}H^{(N_{ij},0)}$, $i=1,2$, $k_i$ are arbitrary natural numbers, and $c_{ij}$ are arbitrary constants, is nothing but the pullback of the bi-Hamiltonian equation
\begin{equation}\label{4}w_t=\check {H}_1\delta_w\check{\mathcal{T}}_1=\check{H}_2\delta_w\check{\mathcal{T}}_2,
\end{equation}
where $\check{H}_i=\sum_{j=1}^{k_i}c_{ij}\check{H}^{(N_{ij},0)}$, $i=1,2$,
$$ \check{H}^{(N,0)}=- D_z\circ\left(\frac{1}{w_z}D_z\right)^{2n},$$
and $\check{\mathcal{T}}_i$ are obtained from $\mathcal{T}_i$ using the substitution $x=z$, $u=w_z$. It is natural to refer to (\ref{4}) as to the potential form of (\ref{3}).

\begin{prp}\label{prlin}
The transformation $z=v$, $w=y$, where $y$ is the new independent variable, linearizes the potential form (\ref{4}) of the bi-Hamiltonian evolution equation (\ref{3}).\looseness=-1
\end{prp}

Before proving this let us point out the following
important consequence of this result.
\begin{ds}\label{ds2}
The differential substitution $x=v$, $u=1/v_y$ relates any equation of the form (\ref{3})
to a linear evolution equation with constant coefficients.
\end{ds}
Informally, this just means that the inverse of the transformation $x=v$,
$u=1/v_y$ linearizes (\ref{3}), but this statement should be treated with some care, as this transformation is not uniquely invertible, and the inverse is not a  differential subsitution, cf.\ Remark~\ref{pz2} and \cite{sokolov}.\looseness=-1

\begin{proof}[Proof of Proposition \ref{prlin}]
The hodograph transformation $z=v$, $w=y$ sends (\ref{4}) into
$$v_t=\widetilde{H}_1\delta_v\widetilde{\mathcal{T}}_1
=\widetilde{H}_2\delta_v\widetilde{\mathcal{T}}_2,$$
where $\widetilde{H}_i=\sum_{j=1}^{k_i}c_{ij}
\widetilde{H}^{(N_{ij},0)}=-\sum_{j=1}^{k_i}c_{ij}D_y^{N_{ij}-2}$ are linear differential operators with constant coefficients,
and $\widetilde{\mathcal{T}}_i$ are obtained from $\check{\mathcal{T}}_i$ using  the transformation in question. 

\begin{lm}\label{lm2}
Let $v_t=K[v]$ be an $n${\em th} order evolution equation which is bi-Hamiltonian with respect to a pair of Hamiltonian operators with constant coefficients. Then $v_t=K[v]$ is
necessarily a {\em linear} equation with constant coefficients, i.e., we have $v_t=\sum_{i=0}^n c_i v_i$, where $c_i=\mathrm{const}$.
\end{lm}
\begin{proof}[Proof of the lemma]
Denote the Hamiltonian operators in question by $\mathcal{D}_i$, $i=1,2$.
Since $v_t=K[v]$ is bi-Hamiltonian with respect to these operators by assumption, their ratio $\mathcal{R}=\mathcal{D}_2\circ\mathcal{D}_1^{-1}$ is a (formal) recursion operator for this equation, that is (see e.g.\ \cite{olver} for details),
$$\mathrm{pr}\ \textbf{v}_K(\mathcal{R})-[\mathrm{D}_K,\mathcal{R}]=0,$$
where $\mathrm{D}_K=\sum_{i=0}^{n}\left(\partial K/\partial u_{i}\right)D^ i$ is the Fr\'echet derivative of $K$ and $\mathrm{pr}\ \textbf{v}_K$ is the prolongation of the evolutionary vector field $\textbf{v}_K$ with the characteristic $K$.
Since $\mathcal{D}_i$, $i=1,2$, have constant coefficients by assumption, we have $\mathrm{pr}\ \textbf{v}_K(\mathcal{D}_i)=0$, and therefore  $\mathrm{pr}\ \textbf{v}_K(\mathcal{R})=0$, so $\mathrm{D}_K$ commutes with $\mathcal{R}$, whence it readily follows that $\mathrm{D}_K$ has constant coefficients (recall that $K$ is independent of $t$ by assumption) and therefore $K$ indeed is a linear combination of $v_i$ with constant coefficients.
\end{proof}
The desired result now readily follows from the above lemma.
\end{proof}

\begin{pr} Consider a bi-Hamiltonian evolution equation
\begin{equation}\label{ex1}
u_t=D_x^3\left(u^{-2}\right)=H^{(3,0)}\delta_u \mathcal{T}_1=H^{(5,0)}\delta_u \mathcal{T}_2,
\end{equation}
where $\mathcal{T}_1=-\int dx/u$ and $\mathcal{T}_2=\int x^2 u dx$.

The potential form (\ref{4}) of (\ref{ex1}) reads
\begin{equation}\label{ex1pot}
w_t=D_z^2\left(w_z^{-2}\right)=\check{H}^{(3,0)}\delta_w\check{\mathcal{T}_1}
=\check{H}^{(5,0)}\delta_w\check{\mathcal{T}_2}.
\end{equation}
Recall that $u=w_z$ and $x=z$; we have $\check{H}^{(3,0)}=-D_z$, $\check{H}^{(5,0)}=-D_z\circ ((1/w_z)D_z)^2$, $\check{\mathcal{T}}_1=-\int dz/w_z$, and $\check{\mathcal{T}}_2=\int z^2 w_z dz$. Note that (\ref{ex1pot}) has, up to a rescaling of $t$, the form (2.31) from \cite{clarkson}.

In perfect agreement with Proposition~\ref{prlin} (cf.\ also Proposition~2.2 in \cite{clarkson}) the hodograph transformation
$z=v$, $w=y$ linearizes (\ref{ex1pot}) into a (trivially) bi-Hamiltonian equation
\begin{equation}\label{ex1lin}
v_t=-2 v_{yyy}=\widetilde{H}^{(3,0)} \delta_v \widetilde{\mathcal{T}}_1=\widetilde{H}^{(5,0)} \delta_v \widetilde{\mathcal{T}}_2,
\end{equation}
where $\widetilde{H}^{(3,0)}=-D_y$, $\widetilde{H}^{(5,0)}=-D_y^3$,
$\widetilde{\mathcal{T}}_1=-\int v_y^2 dy$, and $\widetilde{\mathcal{T}}_2= \int v^2 dy$.  

The transformation $x=v$, $u=1/v_y$ relates (\ref{ex1lin}) to (\ref{ex1}), cf.\ Corollary~\ref{ds2}, so (\ref{ex1}) provides an explicit example of a $C$-integrable (rather than $S$-integrable) bi-Hamiltonian system, just as discussed in Introduction.\looseness=-1
\end{pr}
\section*{Acknowledgements}
The author thanks Dr. A. Sergyeyev for stimulating discussions. This research was supported in part by the Ministry of Education, Youth and Sports of the Czech Republic under grant 
MSM 4781305904.\looseness=-1

\end{document}